\documentclass[11pt]{article}

\usepackage[T5]{fontenc}

\usepackage{amsmath,amssymb,amsthm,algpseudocode,float,caption,mathtools,thmtools,thm-restate,scalerel,stackengine,verbatim}
\usepackage{algorithm,algpseudocode}
  
\stackMath
\newcommand\reallywidehat[1]{%
\savestack{\tmpbox}{\stretchto{%
  \scaleto{%
      \scalerel*[\widthof{\ensuremath{#1}}]{\kern-.6pt\bigwedge\kern-.6pt}%
          {\rule[-\textheight/2]{1ex}{\textheight}}
            }{\textheight}%
            }{0.5ex}}%
            \stackon[1pt]{#1}{\tmpbox}%
            }
\usepackage[margin=1in]{geometry}
\usepackage{multirow}
\usepackage{subcaption}
\usepackage{afterpage}
\usepackage{bm}
\usepackage[hidelinks]{hyperref}
\usepackage{tikz}
\usetikzlibrary{decorations.pathreplacing, positioning}

\usepackage{bbm}
\usepackage{color}

\newcommand{\wt}[1]{\widetilde{#1}}

\newcommand{\abs}[1]{\left|#1\right|}

\DeclarePairedDelimiter\brac{\lbrack}{\rbrack}
\DeclarePairedDelimiter\set{\lbrace}{\rbrace}
\DeclarePairedDelimiter\paren{\lparen}{\rparen}

\newcommand{\E}[2][]{\operatorname*{\mathbb{E}}_{#1 }\brac*{#2}}

\newcommand{\Var}[2][]{\operatorname*{\normalfont{\text{Var}}}_{#1 }\paren*{#2}}
\newcommand{\bO}[1]{\operatorname*{O}\paren*{#1}}
\newcommand{\bOt}[1]{\operatorname*{\wt{O}}\paren*{#1}}
\newcommand{\bOm}[1]{\operatorname*{\Omega}\paren*{#1}}

\newcommand{\bTht}[1]{\operatorname*{\wt{\Theta}}\paren*{#1}}

\newcommand{\fb}{\mathbf{f}}
\newcommand{\gb}{\mathbf{g}}

\newcommand{\hb}{\mathbf{h}}

\newcommand{\obT}{\mathbf{\overline{T}}}

\newcommand{\peq}{\mathrel{+}=}

\newcommand{\return}{\textbf{return}~}

\newcommand{\bool}{{\{0, 1\}}}

\DeclareMathOperator{\poly}{poly}

\DeclareMathOperator{\polylog}{\normalfont{\text{polylog}}}

\newtheorem{theorem}{Theorem}
\newtheorem*{theorem*}{Theorem}
\newtheorem{lemma}[theorem]{Lemma}

\newtheorem*{definition*}{Definition}
\newtheorem*{lemma*}{Lemma}

\newtheorem*{corollary*}{Corollary}

\newtheorem*{claim*}{Claim}

\newtheorem{remark}[theorem]{Remark}

{\makeatletter
	\gdef\xxxmark{%
		\expandafter\ifx\csname @mpargs\endcsname\relax 
		\expandafter\ifx\csname @captype\endcsname\relax 
		\marginpar{xxx}
		\else
		xxx 
		\fi
		\else
		xxx 
		\fi}
	\gdef\xxx{\@ifnextchar[\xxx@lab\xxx@nolab}
	\long\gdef\xxx@lab[#1]#2{{\bf [\xxxmark #2 ---{\sc #1}]}}
	\long\gdef\xxx@nolab#1{{\bf [\xxxmark #1]}}
}

\title{An Optimal Algorithm for Triangle Counting in the Stream}
\date{}

\author{Rajesh Jayaram\\ CMU \\\texttt{rkjayara@cs.cmu.edu} \and John Kallaugher\\UT Austin\\ \texttt{jmgk@cs.utexas.edu}
}

\begin{document}
\maketitle
\begin{abstract}
\noindent
We present a new algorithm for approximating the number of triangles in a graph $G$ whose edges arrive as an arbitrary order stream. If $m$ is the number of edges in $G$, $T$ the number of triangles, $\Delta_E$ the maximum number of triangles which share a single edge, and $\Delta_V$ the maximum number of triangles which share a single vertex, then our algorithm requires space:
\[	\bOt{\frac{m}{T}\cdot \paren*{\Delta_E + \sqrt{\Delta_V}}}
\]
Taken with the $\bOm{\frac{m \Delta_E}{T}}$ lower bound of Braverman,
Ostrovsky, and Vilenchik (ICALP 2013), and the $\bOm{\frac{m
\sqrt{\Delta_V}}{T}}$ lower bound of Kallaugher and Price (SODA 2017), our
algorithm is optimal up to $\log$ factors, resolving the complexity of 
a classic problem in graph streaming.
\end{abstract}

\section{Introduction}
Triangle counting is a fundamental problem in the study of graph algorithms, and one of the best studied in the field of graph streams. It
arises in the analysis of social networks~\cite{BHLP11}, web
graphs~\cite{EM02}, and spam detection~\cite{BBCG08}, among other applications.
From a theoretical perspective, it is of particular interest as the simplest
subgraph counting problem that cannot be solved by considering only \emph{local}
information about individual vertices. In other words, counting triangles
requires one to aggregate information between pairs of \textit{non-incident}
edges. 

In this paper, we present an optimal algorithm for counting triangles in the
\emph{graph streaming} setting, settling a long line of work on this problem. 

\paragraph{Graph Streaming.} In the (insertion-only) graph streaming setting, a
graph $G = (V,E)$ is received as a stream of edges $(\sigma_t)_{t=1}^m$ from
its edge set $E$ in an arbitrary order, and an algorithm is required to output
the answer to some problem at the end of the stream, using as little space as
possible\footnote{Other properties, such as update time, are also of interest,
but space has been the primary object of study in the theory of streaming.}.
Variants on this model include turnstile streaming (in which edges may be
deleted as well as inserted), and models that restrict what kind of state the
algorithm may maintain. 
    
\paragraph{Triangle Counting in Graph Streams.} The theoretical study of graph
streaming was initiated by \cite{BKS02}, who studied the problem of triangle
counting---the problem of estimating the number of three-cliques in a graph.
They demonstrated that, in general, sublinear space algorithms cannot exist for
this problem; namely, in the worst case any algorithm for triangle counting in
a stream must use $\bOm{n^2}$ bits of space.  On the other hand, they also
showed that, if one parameterizes in terms of the number of triangles $T$, one
can often beat this pessimistic lower bound. In particular, they gave an
algorithm that uses $\bOt{\paren{\frac{mn}{T}}^3}$ space to count triangles in
a graph with $m$ edges, $n$ vertices, and $T$ triangles, based on streaming
algorithms for approximating frequency moments.\footnote{Here we assume the
desired approximation is a multiplicative $(1 \pm \varepsilon)$ with success
probability $\delta$ for some positive constants $\varepsilon, \delta$.  For
most algorithms mentioned here, including our own, the dependence on
non-constant $\varepsilon, \delta$ will go as
$\varepsilon^{-2}\log\delta^{-1}$. We use $\bOt{\cdot}$ to suppress
logarithmic or polylogarithmic factors in the argument.}  Of course, it is
unreasonable
to assume that an algorithm knows the number of triangles $T$ in advance, as
this would make counting superfluous. Instead, it will suffice to have constant
factor bounds on the parameters in question.\footnote{One might hope to use
these parameters adaptively, giving an algorithm that uses more space the
smaller $T$ is without needing a lower bound at the start. However, this is in
general impossible, as a graph stream with few triangles and a graph stream
with many triangles may be indistinguishable until the last few updates.}

Several years later, the upper bound for this problem was improved to
$\bOt{\frac{mn}{T}}$ by \cite{BFLMS06}, while \cite{JG05} gave a
(non-comparable) algorithm that samples edges and stores neighborhoods of their
endpoints in order to find triangles, achieving $\bOt{\frac{md^2}{T}}$ space in
graphs with maximum degree $d$. Both algorithms were later subsumed by the
$\bOt{\frac{md}{T}}$ space algorithm of \cite{PTTW13}.

\paragraph{Additional Graph Parameters for Triangle Counting.} Despite the
large strides made by the aforementioned algorithms, none of them can achieve
sublinear space, even for graphs guaranteed to have as many as $\Omega(m)$
triangles, without bounding parameters of the graph other than $m$ and $T$.
This feature was shown to be necessary by \cite{BOV13}, who constructed a
family of graphs with either $0$ or $\bOm{m}$ triangles such that
distinguishing between the two requires $\bOm{m}$ space. However, this ``hard
instance'' is an unusual graph---every triangle in it shares a single edge.
This motivated the introduction of a new graph parameter $\Delta_E$, defined as
the maximum number of triangles which share a single edge in $G$.  When one
parameterizes in terms of $\Delta_E$, the lower bound of \cite{BOV13} becomes
$\bOm{\frac{m\Delta_E}{T}}$. As it happens, the maximum degree of graphs in
this family is also $\Delta_E$, so in particular this proves \cite{PTTW13} to
be optimal among algorithms parametrized by only $m, d$, and $T$.

The first algorithm to directly take advantage of the new parameter $\Delta_E$
was given by \cite{TKMF09}. Their algorithm is simple: keep each edge in the
stream independently with probability $p$, count the number of triangles $T'$
in the resulting graph, and output $T'p^{-3}$. They show that setting $p =
\bO{\frac{1}{T^{1/3}} + \frac{\Delta_E}{T}}$ suffices for an accurate count,
and thereby achieve $\bOt{m\paren*{\frac{1}{T^{1/3}} + \frac{\Delta_E}{T}}}$
space. 

This algorithm has another important feature: it is a \emph{non-adaptive
sampling} algorithm---whether it keeps an edge it sees does not depend on the
contents of the stream before the edge arrives. This means it can naturally
handle \emph{turnstile streams}, streams in which edges may be deleted as well
as inserted. In fact, through the use of sketches for $\ell_0$ sampling (see
e.g.~\cite{CJ19}) such algorithms may be converted into \emph{linear sketches},
which are algorithms that store only a linear function of their input (when
considered as a vector in $\bool^{|V| \choose 2}$).

An improved non-adaptive sampling algorithm was given in \cite{PT12}, which used
the technique of coloring vertices with one of $k$ colors, and keeping all
monochromatic edges. This improved the space usage of the algorithm to
$\bOt{m\paren*{\frac{1}{\sqrt{T}} + \frac{\Delta_E}{T}}}$. In \cite{KP17}, it
was shown (in combination with the existing lower bound of \cite{BOV13}) that
this is optimal, even for insertion-only algorithms---for every $T$ up to
$\Omega(m)$, a family of graphs exist with $\Delta_E \le 1$ and either $0$ or
$T$ triangles, such that $\bOm{\frac{m}{\sqrt{T}}}$ space is required to
distinguish the two.

However, as with the lower bound of \cite{BOV13}, the hard instance from
\cite{KP17} is a rather strange graph: this time every triangle shares a single
\emph{vertex}. Also similarly to the lower bound of \cite{BOV13}, the bound
from \cite{KP17} weakens as the maximum number of triangles sharing a single
vertex, a parameter denoted by $\Delta_V$, is restricted. In this case, when
parameterized by $\Delta_V$, the lower bound becomes
$\bOm{\frac{m\sqrt{\Delta_V}}{T}}$. This was accompanied in \cite{KP17} by an
algorithm that achieves $\bOt{m\paren*{\frac{1}{T^{2/3}} +
\frac{\sqrt{\Delta_V}}{T} + \frac{\Delta_E}{T}}}$ space, improving on
\cite{PT12} for graphs with $\Delta_V = o(T)$.

Subsequently, it was shown in \cite{KKP18} that any linear sketching algorithm
for counting triangles requires $\bOm{\frac{m}{T^{2/3}}}$ space, even if every
triangle is disjoint from every other and therefore $\Delta_E = \Delta_V \le
1$, and so the \cite{KP17} algorithm is optimal among linear sketches. By the
turnstile streaming-linear sketching equivalence of \cite{LNW14}, this suggests
that \cite{KP17} is also optimal among turnstile streaming
algorithms.\footnote{However, the \cite{LNW14} equivalence depends on rather
stringent conditions that a turnstile algorithm must satisfy. In \cite{KP20},
it was shown that relaxing these conditions allows turnstile streaming
algorithms for triangle counting that are closer to the result of \cite{JG05}.}

However, this leaves open the question of how hard triangle counting is for
algorithms that are \emph{not} required to handle deletions (i.e., the standard
``insertion-only'' model). We resolve this question (up to a log factor, as
with previous optimality results), by giving an optimal algorithm for triangle
counting in insertion-only streams.

\begin{figure}
      \centering
      \renewcommand{\arraystretch}{1.5} 
      \begin{tabular}{|l|l|l|} 
        \hline
        Paper & Space & Model\\
        \hline\hline

        \cite{PTTW13} & $\bOt{\frac{md}{T}}$ &      Insertion-only\\ 
        \cite{BOV13} & $\bOm{\frac{m\Delta_E}{T}}$ & 		Insertion-only\\ 
        \cite{KP17} & $\bOm{\frac{m\sqrt{\Delta_V}}{T}}$ &     Insertion-only\\ 
        \hline
        \cite{PT12} & $\bOt{m\paren*{\frac{1}{\sqrt{T}} +
        \frac{\Delta_E}{T}}}$ & Linear Sketching\\ 
        \cite{KP17} & $\bOt{m\left(\frac{1}{T^{2/3}} +
        \frac{\sqrt{\Delta_V}}{T} + \frac{\Delta_E}{T}\right)}$ & Linear
        Sketching\\
        \cite{KKP18} & $\bOm{\frac{m}{T^{2/3}}}$ & Linear Sketching\\
        \hline
        This work & $\bOt{\frac{m}{T}\paren{\sqrt{\Delta_V} + \Delta_E}}$ &
        Insertion-only\\
        \hline
      \end{tabular}
      \caption{Best known upper and lower bounds for triangle counting for
      insertion-only and linear sketching algorithms. $m$ is the number of
      edges, $T$ the number of triangles, $d$ the maximum degree, and
      $\Delta_E$, $\Delta_V$ are the maximum number of triangles sharing an
      edge or a vertex respectively.  Note that linear sketching upper bounds
      imply insertion-only upper bounds, while lower bounds are the opposite.}
      \label{fig:existing}
\end{figure}

\paragraph{Our Algorithm.} We give a new algorithm for counting triangles in insertion-only graph streams.
\begin{restatable}{thm}{optimaltcalg}
\label{thm:optimaltcalg}
For every $\varepsilon, \delta \in (0,1)$, there is an algorithm for
insertion-only graph streams that approximates the number of triangles in a
graph $G$ to $\varepsilon T$ accuracy with probability $1 - \delta$, using \[
\bO{\frac{m}{T}\paren*{\Delta_E
 + \sqrt{\Delta_V}}\log n\frac{\log\frac{1}{\delta}}{\varepsilon^2}}
\]  
bits of space, where $m$ is the number of edges in $G$, $T$ the number of
triangles, $\Delta_E$ the maximum number of triangles which share a single
edge, and $\Delta_V$ the maximum number of triangles which share a single
vertex.
\end{restatable}
\noindent
This matches, up to a log factor (and for constant $\varepsilon,
\delta$), the lower bounds of \cite{BOV13} and \cite{KP17}. It subsumes both
the algorithm of \cite{KP17} and the $\bOt{\frac{md}{T}}$ algorithm of
\cite{PTTW13}, as in any graph with max degree $d$, we have $\Delta_E \le d$
and $\Delta_V \le {d \choose 2}$. This closes the line of work discussed above
on the complexity of triangle counting in insertion-only streams. 

\paragraph{Other Related Work}
In the \emph{multi-pass} streaming setting, an algorithm is allowed to pass
over the input stream more than once. \cite{CJ14} shows multipass algorithms
take $\bTht{m/\sqrt{T}}$ space for arbitrary graphs, giving an algorithm for
two passes and a lower bound for a constant number of passes.  \cite{KMPT12}
shows a three pass streaming algorithm using $O(\sqrt{m} + m^{3/2}/T)$ space.
\cite{BC17} gave a $\bO{m^{3/2}/T}$ four pass algorithm. 

In the \emph{adjacency-list} model, in which each vertex's list of neighbors is
received as a block (and so in particular every edge is seen twice),
\cite{MVV16} gave a $\bO{m/\sqrt{T}}$ space one-pass algorithm,
while~\cite{KMPV19} gave $\bO{m/T^{2/3}}$ space 2-pass algorithm, as well as
tight (but conditional on open communication complexity conjectures) lower
bounds for both.

The problem has also been studied in the query model, in which case rather than
space the concern is minimizing time or query count. While this is a very
different setting, similar concerns around mitigating the impact of ``heavy''
vertices or edges arise. \cite{ELRS15} considered triangle
counting in this setting, which was extended by~\cite{ERS18} to general cliques
and~\cite{AKK19} to arbitrary constant-size subgraphs.

\section{Overview of the Algorithm}
At a high-level, many triangle counting algorithms in the literature adhere to
the following template: \textbf{(1)} design a sampling scheme to sample
triangles, \textbf{(2)} count the number of triangles which survive after this
sampling process, \textbf{(3)} rescale the number of empirically sampled
triangles by the expected fraction of surviving triangles to obtain an unbiased
estimator for $T$. 

As an example, one could sample each edge uniformly with probability $q$ (this
is the approach taken in \cite{TKMF09}).  Since for a triangle to survive all
three of its edges must be sampled, the expected number of triangles that
survive is $Tq^3$. Thus, rescaling the number of empirically sampled triangles
by $1/q^3$ yields an unbiased estimator. How large must $q$ be to make this
estimator accurate? In order to sample even a single triangle we need $Tq^3
\geq 1$, so clearly $q$ must be at least $1/T^{1/3}$. Moreover, if $\Delta_E$ is
the largest number of triangles that share an edge, there might be as few as
$T/\Delta_E$ ``heavy'' edges such that sampling a triangle requires sampling at
least one of them, and so $q$ must be at least $\Delta_E/T$. It turns out that,
up to constant factors, this is also sufficient, and so the space
needed by this algorithm is $\bOt{m\paren*{\frac{1}{T^{1/3}} +
\frac{\Delta_E}{T}}}$ bits.

The starting point for our algorithm is the following simple observation, which
can be seen as an optimization to the sampling algorithm above.  Given three
edges $uv,vw,wu \in E$ arriving in a stream in that order, once the first two
edges $uv,vw$ have been sampled and stored, upon seeing the ``completing''
edge $wu$, we will know that the triangle $uvw$ exists in $G$, and may count it
immediately---we get the closing edge of each triangle ``for free''. Now for a
single triangle to be sampled, we only need to sample the first two edges, and
so the probability of finding any given triangle improves to $q^2$, allowing a
space complexity of $\bOt{m\paren*{\frac{1}{\sqrt{T}} + \frac{m
\Delta_E}{T}}}$. However, when $\Delta_V = o(T)$, this is still weaker than
allowed by the $\bOm{\frac{m}{T}\paren{\sqrt{\Delta_V} + \Delta_E}}$ lower
bound that results from combining the results of \cite{BOV13,KP17}.

While the aforementioned algorithm is sub-optimal in general, notice that it
does match the lower bounds in the extreme case when $\Delta_V =T$, and all
triangles share a single vertex.  On the other hand, when $\Delta_V$ is
smaller, there are more `fully disjoint'' triangles in the graph.
Consequentially, we can afford to subsample by \emph{vertices}, as now dropping
a single vertex cannot lose too large a fraction of our triangles.   We may
sample vertices uniformly with some probability $p$, and deterministically
store all edges adjacent to at least one sampled vertex, again counting a
triangle whenever we observe an edge $wu$ closing a sampled pair $uv$, $vw$.
Each such triangle will be counted iff the ``first'' vertex $v$ of the triangle
is sampled, and these may be divided among as few as $T/\Delta_V$ ``heavy''
vertices, so $p$ must be at least $\Delta_V/T$.  This again turns out to be
sufficient, for a space usage of $\bOt{\frac{m \Delta_V}{T}}$ (note that any
pair of edges sharing an edge also share a vertex, so $\Delta_E \le \Delta_V$,
and thus this does not violate the known lower bounds). While this is an
improvement on the aforementioned adaptive edge-sampling scheme for small
$\Delta_V$, it becomes worse once $\Delta_V > \sqrt{T}$. 

The crucial insight behind our algorithm is to merge the two aforementioned
algorithms with a careful choice of parameterization. Specifically, we sample
both edges \textit{and} vertices, before counting triangles that we see closing
our sampled wedges. Specifically, we sample vertices $v \in V$ in the graph
with probability $p \in (0,1\rbrack$, and then ``activate'' each edge $e \in E$
with probability $q \in (0,1\rbrack$. When an edge $uv \in E$ arrives in the
stream, we store it iff $uv$ is active \textit{and} at least one of the
vertices $u$ or $v$ was sampled. We denote by $S$ the set of all edges stored
by the algorithm. Finally, when a closing edge $wu$ arrives that completes a
triangle with edges $uv,vw$ that were previously added to $S$, we check if the
vertex $v$ at the center of the wedge $uv,vw$ was sampled, and if so we
deterministically increment a counter $\mathbf{C}$.

Now observe that, for any given triangle $uvw$, the probability that $uvw$
causes $\mathbf{C}$ to be incremented is exactly $pq^2$. Thus, if we output the
quantity $\mathbf{C}/(pq^2)$ at the end of the stream, we obtain an unbiased
estimator for the number of triangles in $G$. 

Notice that when $p=1$ our algorithm reduces to the simpler edge-sampling
algorithm stated above.  At the other extreme, when $q =1$ our algorithm
reduces to the vertex sampling algorithm. Intuitively, our choice of the
parameters $p$ and $q$ are subject to the same constraints faced by the
aforementioned edge- and vertex-sampling algorithms. Firstly, $p$ must be at
least $\Delta_V/T$, otherwise the algorithm could miss a ``heavy'' vertex.
Furthermore, the product $pq$ must be at least $\Delta_E/T$, to avoid missing
``heavy'' edges, and $pq^2$ must be at least $1/T$ to find any triangles at
all.  Putting these bounds together, it follows that $q$ must be at least
$\max\set*{\frac{\Delta_E}{\Delta_V}, \frac{1}{\sqrt{\Delta_V}}}$.

As with all the algorithms discussed so far, this turns out to also be
sufficient---we demonstrate that by fixing the sampling parameters\footnote{As
mentioned earlier, $\Delta_E \le \Delta_V$, while $\Delta_V \le T$ holds
trivially. Thus $p,q$ are valid probabilities.}  \[
p = \frac{\Delta_V}{T}, \quad \quad\quad \quad  q\geq
\max\left\{\frac{\Delta_E}{\Delta_V}, \frac{1}{\sqrt{\Delta_V}} \right\}\] 
we obtain an algorithm using space
$\bO{\frac{m}{T}\paren*{\Delta_E + \sqrt{\Delta_V}}\log n}$
which yields an $\bO{T^2}$ variance estimator. We may therefore obtain a $(1
\pm \varepsilon)$ multiplicative estimate with probability $1 - \delta$ by
using $\bO{\frac{1}{\varepsilon^2}\log\frac{1}{\delta}}$ copies of this
algorithm.

Consequentially one obtains an algorithm matching, up to a log factor, the
lower bounds of \cite{BOV13,KP17}, with optimal space usage in terms of
$m, T,\Delta_E,\Delta_V$.

\section{The Triangle Counting Algorithm}
Let $G = (V,E)$ be a graph on $n$ vertices, received as a stream of undirected
edges, adversarially ordered. Let $m$ be the number of edges in the stream. We
write the stream as $\sigma = (\sigma_i)_{i=1}^m$, with each $\sigma_i \in E$.
We use $T$ to refer to the number of triangles in $G$, $\Delta_E$ to refer to
the maximum number of them sharing a single edge, and $\Delta_V$ the maximum
number sharing a single vertex. 

\begin{remark}
As with all streaming triangle counting algorithms, our algorithm will need to
be parametrized by statistics of the graph that cannot be known exactly without
trivializing the problem---in our case $T$, $\Delta_E$, and $\Delta_V$.
However, it will not be necessary to know these exactly---an upper bound on
$\Delta_E$, $\Delta_V$ and a lower bound on $T$ will be sufficient. If these
bounds are tight up to a constant, the complexity of our algorithm will be
unchanged, otherwise replace the parameters $T$, $\Delta_E$, $\Delta_V$ with
the respective upper and lower bounds. 
\end{remark}

\subsection{Description of the Algorithm}
We begin by choosing two hash functions $\fb : V \rightarrow \bool$ and $\gb :
E \rightarrow \bool$, which will serve as our ``vertex sampling'' and ``edge
sampling'' functions, respectively. We choose $\fb$ to be pair-wise
independent. $\gb$ will only be evaluated at most once for each edge,  and so
we may choose it to be fully independent. We pick the two functions $\fb,\gb$
such that 
 \[
\E{\fb(v)} = p
\]
for each $v \in V$ and
 \[
\E{\gb(e)} = q
\]
for each $e \in E$, where $p,q$ are parameters to be set later. Such a hash
function $\fb$ can be generated by taking a two-wise independent function
$\hb: V \to [M]$, where $M = \poly(n)$ is a sufficiently large
multiple of $1/p$, and setting $\fb(v) = 1$ whenever $\hb(v) \leq
pM$ (one can construct $\gb$ similarly using a four-wise independent hash
function). Such functions can be generated and stored in at most $O(\log n)$
bits of space \cite{carter1979universal}.

The algorithm will be simple: sample vertices with probability $p$, sample
incident edges with probability $q$. The formal description is given below in Algorithm \ref{alg:octc}.
\begin{algorithm}[H]
\caption{Triangle Counting Algorithm}
\label{alg:octc}
\begin{algorithmic}[1]
\Procedure{TriangleCounting}{$p,q$}
\State $S \gets \emptyset$ 
\State $\obT \gets 0$
\For{each update $wv$}
\For{$u \in V$}
\If{$\fb(u) > 0 \wedge uv, uw \in S$}
\State $\obT \peq 1/pq^2$
\EndIf
\EndFor
\If{$\gb(wv)(\fb(w) + \fb(v)) > 0$}
\State $S \gets S \cup \set{wv}$
\EndIf
\EndFor
\State \return $\obT$.
\EndProcedure
\end{algorithmic}
\end{algorithm}

\subsection{Analysis of the Algorithm}
\begin{lemma}
\label{lm:octcsp}
This algorithm uses $\bO{mpq\log n}$ bits of space.
\end{lemma}
\begin{proof}
Besides an $\bO{\log n}$ sized counter and the hash function $\fb$ ($\gb$ is
never evaluated more than once for an edge and thus does not need to be
stored), the algorithm maintains a set of edges. Each edge will be kept with
probability at most $2pq$ and takes $O(\log n)$ space to store, so the result
follows.
\end{proof}

We will write $T_{uvw}$ for the variable that is $1$ if $uvw$ is a triangle in
$G$ with its edges arriving in the order $(uv,uw,vw)$, and 0 otherwise, and so \[
T = \sum_{(u,v,w) \in V^3} T_{uvw}\text{.}
\] We will
write $\obT_{uvw}$ for the random variable that is $1/pq^2$ if $T_{uvw} = 1$
and $\fb(u)\gb(uv)\gb(uw) = 1$, and $0$ otherwise. We will therefore have \[ 
\obT = \sum_{(u,v,w) \in V^3} \obT_{uvw}.
\]

\begin{lemma}
\label{lm:octcex}
\[
\E{\obT} = T\text{.}
\]
\end{lemma}
\begin{proof}
For any $(u,v,w)$, $\fb(u)\gb(uv)\gb(uw) = 1$ with probability $pq^2$, so
$\E{\obT_{uvw}} = T_{uvw}$. Therefore,
\begin{align*}
\E{\obT} &= \sum_{(u,v,w) \in V^3} \E{\obT_{uvw}}\\
&= \sum_{(u,v,w) \in V^3}T_{uvw}\\
&= T
\end{align*}
\end{proof}

\begin{lemma}
\label{lm:octcvar}
\[
\Var{\obT} \le T/pq^2 + T\Delta_E/pq + T\Delta_V/p\text{.}
\]
\end{lemma}
\begin{proof}
Consider any (ordered) pair of triples $(u,v,w), (x,y,z) \in V^3$ such that
$T_{uvw}T_{xyz} = 1$. 

If $(u,v,w) = (x,y,z)$, $\obT_{uvw}\obT_{xyz} = 1/p^2q^4$ with probability
$pq^2$ and $0$ otherwise, so \[ 
\E{\obT_{uvw}\obT_{xyz}} = \E{\obT_{uvw}^2} = 1/pq^2.
\]
At most $T$ such pairs of triples can exist. 

Now, if $\abs{\set{uv,uw} \cap \set{xy,xz}} = 1$, then $u = x$ and so
$\obT_{uvw}\obT_{xyz} = 1/p^2q^4$ iff $\fb(u) = 1$ and $\gb(e) = 1$ for all $e$
in the size-3 set $\set{uv,uw,xy,xz}$, which happens with probability $pq^3$,
and so \[
\E{\obT_{uvw}\obT_{xyz}} = 1/pq.
\]
Each triangle has at most $\Delta_E$ other triangles it shares an edge with, so
there are at most $T\Delta_E$ such pairs.

If $\set{uv,uw} \cap \set{xy,xz} = \emptyset$ but $u = x$, then
$\obT_{uvw}\obT_{xyz} = 1/p^2q^4$ iff $\fb(u) = 1$ and $\gb(e) = 1$ for all $e$
in the size-4 set $\set{uv,uw,xy,xz}$, which happens with probability $pq^4$,
and so \[
\E{\obT_{uvw}\obT_{xyz}} = 1/p.
\]
Each triangle has at most $\Delta_V$ other triangles it shares a vertex with,
so there are at most $T\Delta_V$ such pairs.

Finally, if $\set{u,v,w} \cap \set{x,y,z} = \emptyset$, then
$\obT_{uvw}\obT_{xyz} = 1/p^2q^4$ iff $\fb(u) = 1$, $\fb(x) = 1$, and $\gb(e) =
1$ for all $e$ in the size-4 set $\set{uv,uw,xy,xz}$, which happens with
probability $p^2q^4$, and so \[ 
\E{\obT_{uvw}\obT_{xyz}} = 1.
\]
At most $T^2$ such pairs can exist. Therefore,
\begin{align*}
\E{\obT^2} &= \sum_{(u,v,w) \in V^3}\sum_{(x,y,z) \in V^3} \E{\obT_{uvw}\obT_{xyz}}\\
&= \sum_{(u,v,w) \in V^3} \E{\obT_{uvw}^2} + \sum_{(u,v,w) \in V^3} 
\left(\sum_{\substack{(x,y,z) \in
V^3\\\abs{\set{uv, uw} \cap \set{xy,xz}} = 1}} \E{\obT_{uvw}\obT_{xyz}} +\right.\\
&\phantom{=}\left. \sum_{\substack{(x,y,z) \in V^3\\\set{uv, uw} \cap
\set{xy,xz} = \emptyset\\ u = x}} \E{\obT_{uvw}\obT_{xyz}} +
\sum_{\substack{(x,y,z) \in V^3\\\set{u,v,w} \cap \set{x,y,z} = \emptyset}}
\E{\obT_{uvw}\obT_{xyz}} \right)\\
&\le T/pq^2 + T\Delta_E/pq + T\Delta_V/p + T^2
\end{align*}
by adding the previously established bounds for all four kinds of pair. The
lemma then follows from the fact that $\Var{\obT} = \E{\obT^2} - \E{\obT}^2 =
\E{\obT^2} - T^2$.
\end{proof}

We may now prove Theorem~\ref{thm:optimaltcalg}.
\optimaltcalg*
\begin{proof}
We may assume $\Delta_V$ (more specifically, the upper bound we have on it) is
at least 1, as otherwise we already know $G$ to be triangle-free. By
Lemmas~\ref{lm:octcex} and~\ref{lm:octcvar}, we can set $p = \Delta_V/T$, $q =
\max\set*{\Delta_E/\Delta_V, 1/\sqrt{\Delta_V}}$ and run
Algorithm~\ref{alg:octc} to obtain an estimator with expectation $T$ and
variance at most $3T^2$. (These will give valid probabilities, as $\Delta_V \le
T$ by definition, and $\Delta_E$ is at least $\Delta_V$, as any pair of
triangles sharing an edge also share a vertex.) By Lemma~\ref{lm:octcsp}, this
will take $\bO{\frac{m}{T}\paren*{\Delta_E + \sqrt{\Delta_V}}\log
n}$ space. 

Repeating this $36/\varepsilon^2$ times and taking the mean will give an
estimator with expectation $T$ and variance at most $\varepsilon T^2/2$. We can
then repeat \emph{this} $\bO{\log \frac{1}{\delta}}$ times and take the median
to get an estimator that will be within $\varepsilon T$ of $T$ with probability
$1 - \delta$.
\end{proof}

\section{Conclusion}
We resolve the complexity of triangle counting in the insertion-only streaming
model, in terms of the well-studied natural graph parameters $m, T, \Delta_E,
\Delta_V$.  The results of~\cite{KKP18} resolved this problem for the
\emph{linear sketching} model, and a result of~\cite{LNW14} states that, under
certain conditions, turnstile streaming algorithms are equivalent to linear
sketches, suggesting that the algorithm of~\cite{KP17} is optimal for turnstile
streams as well. However,~\cite{KP20} showed that an insertion-only algorithm
of~\cite{JG05} can be converted into a turnstile streaming algorithm provided
that, for instance, the length of the stream is reasonably constrained (with
the number of insertions and deletions no more than $O(1)$ times the final size
of the graph). It remains open whether this algorithm can be converted into a
turnstile algorithm under such constraints, or whether the bounded-stream
turnstile complexity of triangle counting is somewhere between insertion-only
and linear sketching.

Another natural question is about the choice of parameters---the algorithm of
\cite{PT12} is optimal in terms of $m, T$, and $\Delta_E$, but not when the
parameter $\Delta_V$ is considered. Are there natural extensions of the
parametrization that allow for better results? The results of \cite{KP17}
include a proof of instance-optimality for a restricted subclass of
non-adaptive sampling algorithms, but for more general algorithms it is clear
that there are at least \emph{unnatural} extensions of the parametrization that
help. For instance, if all the edges of a graph are guaranteed to belong to
high-degree vertices, but all the triangles belong to low-degree vertices, a
simple filtering strategy allows an improvement.

In particular, the lower bound instances of \cite{BOV13, KP17} are both sparse
graphs, and so cannot be constructed if $n$ is constrained to be small relative
to $m$ or $T$. For the most dense graphs (with $\Theta\paren*{n^2}$ edges and
$\Theta\paren*{n^3}$ triangles) our algorithm and the algorithm of~\cite{KP17}
are already trivially optimal up to log factors, since they use only
$\polylog(n)$ bits. However, the complexity landscape for more general dense
graphs remains open.

\bibliographystyle{alpha}
\bibliography{refs}

\newcommand{\etalchar}[1]{$^{#1}$}
\begin{thebibliography}{PTTW13}

\bibitem[AKK19]{AKK19}
Sepehr Assadi, Michael Kapralov, and Sanjeev Khanna.
\newblock A simple sublinear-time algorithm for counting arbitrary subgraphs
  via edge sampling.
\newblock In Avrim Blum, editor, {\em 10th Innovations in Theoretical Computer
  Science Conference, {ITCS} 2019, January 10-12, 2019, San Diego, California,
  {USA}}, volume 124 of {\em LIPIcs}, pages 6:1--6:20. Schloss Dagstuhl -
  Leibniz-Zentrum f{\"{u}}r Informatik, 2019.

\bibitem[BBCG08]{BBCG08}
Luca Becchetti, Paolo Boldi, Carlos Castillo, and Aristides Gionis.
\newblock Efficient semi-streaming algorithms for local triangle counting in
  massive graphs.
\newblock In {\em Proceedings of the 14th ACM SIGKDD International Conference
  on Knowledge Discovery and Data Mining}, KDD '08, pages 16--24, New York, NY,
  USA, 2008. ACM.

\bibitem[BC17]{BC17}
Suman~K. Bera and Amit Chakrabarti.
\newblock {Towards Tighter Space Bounds for Counting Triangles and Other
  Substructures in Graph Streams}.
\newblock In {\em 34th Symposium on Theoretical Aspects of Computer Science
  (STACS 2017)}, volume~66 of {\em Leibniz International Proceedings in
  Informatics (LIPIcs)}, pages 11:1--11:14, Dagstuhl, Germany, 2017. Schloss
  Dagstuhl--Leibniz-Zentrum fuer Informatik.

\bibitem[BFL{\etalchar{+}}06]{BFLMS06}
Luciana~S Buriol, Gereon Frahling, Stefano Leonardi, Alberto
  Marchetti-Spaccamela, and Christian Sohler.
\newblock Counting triangles in data streams.
\newblock In {\em Proceedings of the twenty-fifth ACM SIGMOD-SIGACT-SIGART
  symposium on Principles of database systems}, pages 253--262. ACM, 2006.

\bibitem[BHLP11]{BHLP11}
Jonathan~W. Berry, Bruce Hendrickson, Randall~A. LaViolette, and Cynthia~A.
  Phillips.
\newblock Tolerating the community detection resolution limit with edge
  weighting.
\newblock {\em Phys. Rev. E}, 83:056119, May 2011.

\bibitem[BKS02]{BKS02}
Ziv {Bar-Yossef}, Ravi Kumar, and D.~Sivakumar.
\newblock Reductions in streaming algorithms, with an application to counting
  triangles in graphs.
\newblock In {\em Proceedings of the Thirteenth Annual ACM-SIAM Symposium on
  Discrete Algorithms}, SODA '02, pages 623--632, Philadelphia, PA, USA, 2002.
  Society for Industrial and Applied Mathematics.

\bibitem[BOV13]{BOV13}
Vladimir Braverman, Rafail Ostrovsky, and Dan Vilenchik.
\newblock How hard is counting triangles in the streaming model?
\newblock In {\em Automata, Languages, and Programming}, pages 244--254.
  Springer, 2013.

\bibitem[CJ14]{CJ14}
Graham Cormode and Hossein Jowhari.
\newblock A second look at counting triangles in graph streams.
\newblock {\em Theoretical Computer Science}, 552:44--51, 2014.

\bibitem[CJ19]{CJ19}
Graham Cormode and Hossein Jowhari.
\newblock {$L_p$} samplers and their applications: A survey.
\newblock {\em ACM Comput. Surv.}, 52(1), February 2019.

\bibitem[CW79]{carter1979universal}
J~Lawrence Carter and Mark~N Wegman.
\newblock Universal classes of hash functions.
\newblock {\em Journal of computer and system sciences}, 18(2):143--154, 1979.

\bibitem[ELRS15]{ELRS15}
Talya Eden, Amit Levi, Dana Ron, and C.~Seshadhri.
\newblock Approximately counting triangles in sublinear time.
\newblock In {\em Proceedings of the 56th FOCS}, pages 614--633. IEEE, 2015.

\bibitem[EM02]{EM02}
Jean-Pierre Eckmann and Elisha Moses.
\newblock Curvature of co-links uncovers hidden thematic layers in the world
  wide web.
\newblock {\em Proceedings of the National Academy of Sciences},
  99(9):5825--5829, 2002.

\bibitem[ERS18]{ERS18}
Talya Eden, Dana Ron, and C.~Seshadhri.
\newblock On approximating the number of k-cliques in sublinear time.
\newblock In {\em Proceedings of the 50th Annual ACM SIGACT Symposium on Theory
  of Computing}, STOC 2018, page 722–734, New York, NY, USA, 2018.
  Association for Computing Machinery.

\bibitem[JG05]{JG05}
Hossein Jowhari and Mohammad Ghodsi.
\newblock New streaming algorithms for counting triangles in graphs.
\newblock In {\em Computing and Combinatorics}, pages 710--716. Springer, 2005.

\bibitem[KKP18]{KKP18}
John Kallaugher, Michael K10.1145/3188745.3188810oapralov, and Eric Price.
\newblock The sketching complexity of graph and hypergraph counting.
\newblock In {\em 2018 IEEE 59th Annual Symposium on Foundations of Computer
  Science (FOCS)}, pages 556--567. IEEE, 2018.

\bibitem[KMPT12]{KMPT12}
Mihail~N Kolountzakis, Gary~L Miller, Richard Peng, and Charalampos~E
  Tsourakakis.
\newblock Efficient triangle counting in large graphs via degree-based vertex
  partitioning.
\newblock {\em Internet Mathematics}, 8(1-2):161--185, 2012.

\bibitem[KMPV19]{KMPV19}
John Kallaugher, Andrew McGregor, Eric Price, and Sofya Vorotnikova.
\newblock The complexity of counting cycles in the adjacency list streaming
  model.
\newblock In Dan Suciu, Sebastian Skritek, and Christoph Koch, editors, {\em
  Proceedings of the 38th {ACM} {SIGMOD-SIGACT-SIGAI} Symposium on Principles
  of Database Systems, {PODS} 2019, Amsterdam, The Netherlands, June 30 - July
  5, 2019}, pages 119--133. {ACM}, 2019.

\bibitem[KP17]{KP17}
John Kallaugher and Eric Price.
\newblock A hybrid sampling scheme for triangle counting.
\newblock In {\em Proceedings of the Twenty-Eighth Annual ACM-SIAM Symposium on
  Discrete Algorithms}, pages 1778--1797. SIAM, 2017.

\bibitem[KP20]{KP20}
John Kallaugher and Eric Price.
\newblock Separations and equivalences between turnstile streaming and linear
  sketching.
\newblock In {\em Proccedings of the 52nd Annual {ACM} {SIGACT} Symposium on
  Theory of Computing, {STOC} 2020, Chicago, IL, USA, June 22-26, 2020}, pages
  1223--1236. {ACM}, 2020.

\bibitem[LNW14]{LNW14}
Yi~Li, Huy~L. N{\fontencoding{T5}\selectfont guy\~{\ecircumflex{}}n}, and
  David~P. Woodruff.
\newblock Turnstile streaming algorithms might as well be linear sketches.
\newblock In {\em Symposium on Theory of Computing, {STOC} 2014, New York, NY,
  USA, May 31 - June 03, 2014}, pages 174--183, 2014.

\bibitem[MVV16]{MVV16}
Andrew McGregor, Sofya Vorotnikova, and Hoa~T. Vu.
\newblock Better algorithms for counting triangles in data streams.
\newblock In {\em Proceedings of the 35th ACM SIGMOD-SIGACT-SIGAI Symposium on
  Principles of Database Systems}, PODS '16, pages 401--411, New York, NY, USA,
  2016. ACM.

\bibitem[PT12]{PT12}
Rasmus Pagh and Charalampos~E Tsourakakis.
\newblock Colorful triangle counting and a mapreduce implementation.
\newblock {\em Information Processing Letters}, 112(7):277--281, 2012.

\bibitem[PTTW13]{PTTW13}
A.~Pavan, Kanat Tangwongsan, Srikanta Tirthapura, and Kun-Lung Wu.
\newblock Counting and sampling triangles from a graph stream.
\newblock {\em Proc. VLDB Endow.}, 6(14):1870--1881, September 2013.

\bibitem[TKMF09]{TKMF09}
Charalampos~E Tsourakakis, U~Kang, Gary~L Miller, and Christos Faloutsos.
\newblock Doulion: counting triangles in massive graphs with a coin.
\newblock In {\em Proceedings of the 15th ACM SIGKDD international conference
  on Knowledge discovery and data mining}, pages 837--846. ACM, 2009.

\end{thebibliography}

\end{document}